\documentclass[12pt]{iopart}

\usepackage{iopams}  

\usepackage{graphicx,color}
\usepackage{xspace}
\usepackage{bm} 
\usepackage{dsfont}

\usepackage{amsthm}

\definecolor{dgreen}{rgb}{0,0.5,0}

\def\cA{{\mathcal{A}}}

\def\cS{{\mathcal{S}}}

\def\t#1{\tilde{#1}}

\def \L{\mathcal{L}}

\newcommand{\LJB}{\mbox{LJB--algebra}\xspace}
\newcommand{\LJBs}{\mbox{LJB--algebras}\xspace}
\newcommand{\CA}{C^*\mbox{--algebra}\xspace}
\newcommand{\CAs}{C^*\mbox{--algebras}\xspace}

\newtheorem{theorem}{Theorem}
\newtheorem{definition}{Definition}

\newtheorem{lem}{Lemma}

\theoremstyle{definition}

\theoremstyle{remark} 
\newtheorem*{rem}{Remark}

\begin{document}
\title[Quantumness via Jordan product]{Defining quantumness via the Jordan product}
\author{Paolo Facchi$^{1,2}$, Leonardo Ferro$^{3,4,5}$, Giuseppe Marmo$^{3,4}$ and  Saverio Pascazio$^{1,2}$}
\address{$^{1}$Dipartimento di Fisica and MECENAS, Universit\`a di Bari, I-70126  Bari, Italy}
\address{$^{2}$INFN, Sezione di Bari, I-70126 Bari, Italy}
\address{$^{3}$Dipartimento di Fisica and MECENAS, Universit\`a di Napoli ``Federico II'', 
I-80126 Napoli, Italy }
\address{$^{4}$INFN, Sezione di Napoli, I-80126 Napoli, Italy}
\address{$^{5}$Departamento de Matem\'aticas, Universidad Carlos III de Madrid, 
28911 Legan\'es, Madrid, Spain}
\date{\today}

\begin{abstract}
We propose alternative definitions of classical states and quantumness witnesses by focusing on the algebra of observables of the system. A central role will be assumed by the anticommutator of the observables, namely the Jordan product. This approach turns out to be suitable for generalizations to infinite dimensional systems. We then show that the whole algebra of observables can be generated by three elements by repeated application of the Jordan product.
\end{abstract}



\maketitle

\section{Introduction and motivations}

The definition and characterization of the quantum features of a physical system is an important problem, both for its fundamental implications and its practical aspects, due to the advent of quantum information processing~\cite{NC}. There are tasks in computation and communication that can be efficiently performed only if quantum resources are available. Classical systems are not suitable for such applications.

A recent attempt to define and discriminate quantumness and classicality was made in Refs.\ \cite{AVR,APVR} by Alicki and collaborators, who introduced the idea of ``quantumness witness'', motivating interesting experiments \cite{brida08,bridagisin}. These experiments checked the quantumness of the system investigated (a photon) and ruled out possible (semi)classical descriptions.

These studies, as well as those that ensued \cite{quantumness12,PhTrRoyal,anticommutators12}, focused on finite-dimensional systems. In brief, the approach was the following.
Consider a physical system and its $\CA$  $\mathcal{A}$~\cite{Araki, BratteliRobinson}. The key observation \cite{AVR,APVR} is to notice that the following two statements are equivalent: given any pair $a, b \in\mathcal{A}$,
\begin{eqnarray}
i) &  & \mathcal{A} \; \textrm{is commutative:} \; [a,b]=\frac{i}{2} (ab - ba)=0 ; \label{classicalC} 
\label{eq:commutator}
\\
ii) \;&  & a\geq 0, b\geq 0 \; \longrightarrow \; \{a,b\}= ab + ba \geq 0 . \label{classicalA} 
\end{eqnarray}
[We introduced an additional factor $i/2$ in (\ref{classicalC}) with respect to the familiar definition of commutator. In this way the space of real elements acquires the structure of a Lie
algebra.]
Therefore, if the symmetrized product (\ref{classicalA}) of two positive observables can take negative values, the algebra $\mathcal{A}$ is non-Abelian and the system is quantum. A state $\rho$ of a quantum (or classical) system is defined to be classical \cite{quantumness12} if 
\begin{equation}
\label{commab}
 \rho([a,b]) = 0, \qquad \forall  a, b \in\mathcal{A}.
\end{equation}
Classical states are therefore ``transparent" to all commutators and in this sense do not ``detect" the non-commutativity of the algebra.
An observable $q=\{a,b\}$, with $a,b$ positive, is positive on all classical states, but can take negative values for quantum states, ``witnessing" in this way the ``quantumness" of the system.

In the above definitions, the focus is on the whole $\CA$ which is assumed to embed the real algebra of observables.  
In this article we will propose alternative definitions that directly refer to the algebra of  observables and can be generalized to infinite-dimensional systems.

\section{Lie--Jordan algebras}

Let $\cA$ be an associative $*-$algebra. We denote by $\L$ the (real) self-adjoint part of $\cA$, i.e.\ the \emph{algebra of observables}
\begin{equation}
 \L = \{\, a \in \cA \mid a^* = a\,\},
\end{equation}
which, when $\cA$ is noncommutative, is not closed under the associative product:
\begin{equation}
 (ab)^* = ba \neq ab,
\end{equation}
for some $a,b \in \L$. This requires  the introduction of new algebraic structures. By following the seminal ideas by Jordan \cite{Jordan}, later developed in conjunction with  Wigner and von Neumann \cite{JWV34}, we define the Jordan product \cite{Hanche84} as the symmetrized product
\begin{equation}\label{jordan}
a \circ b = \frac{1}{2} (ab + ba) = \frac{1}{2} \{a,b\},\qquad \forall\,a, b \in \L.
 \end{equation}
The self-adjoint algebra $\L$ is closed with respect to this commutative product
\begin{equation}
 (a \circ b)^* = a \circ b,
\end{equation}
but it is \emph{not} associative. However it satisfies a weak form of associativity
\begin{equation}\label{jassociativity}
(a^2 \circ b) \circ a = a^2 \circ (b \circ a) ,
\end{equation}
where \begin{equation}
a^2 = a \circ a. 
\end{equation}
The algebra $\L$ is also naturally endowed with a Lie product 
\begin{equation}
\label{lie}
 \left[a,b\right] = \frac{i}{2} (ab - ba),\qquad \forall\,a, b \in \L ,
 \end{equation}
where the factor 1/2 is introduced for convenience and the imaginary unit in order to make the product self-adjoint:
 \begin{equation}
[a,b]^* = [a,b]\in \L. 
\end{equation}
This antisymmetric bilinear form equips $\L$ with a Lie algebra structure, since it verifies the Jacobi identity:
\begin{equation}\label{jacobi}
 \left[\left[a,b\right],c\right] + \left[\left[c,a\right],b\right] + \left[\left[b,c\right],a\right] = 0.
\end{equation}
Note that these two operations are compatible in the sense that the Leibniz identity is satisfied
\begin{equation}\label{leibniz}
 \left[a,b\circ c\right] = \left[a,b\right]\circ c + b\circ \left[a,c\right] .
\end{equation}
Moreover, the lack of associativity of the Jordan product is related to the Lie bracket by
\begin{equation}
\label{associator}
 (a\circ b)\circ c - a \circ (b \circ c) 
 = \left[b,\left[c,a\right]\right] .
\end{equation}
The above identity will be central in our analysis. It offers a remarkable physical insight. The lack of commutativity (right-hand side) turns out to be related to the symmetric product in a way that is different (and as we will argue, more general) from Eqs.\ (\ref{classicalC})-(\ref{classicalA}). This observation will enable us to define the quantumness of a physical system in terms of the lack of associativity of its Jordan product.

A real vector space $\L$ with a symmetric operation $\circ$ and an antisymmetric one $\left[\cdot ,\cdot \right]$, satisfying properties (\ref{jacobi}), (\ref{leibniz}) and (\ref{associator}) will be called a Lie--Jordan algebra.

\begin{rem}
 If we are given a Lie--Jordan algebra $\L$, we can always recover the associative product on the complexified vector space $\L^\mathbb{C} = \L \oplus i \L$ by defining
\begin{equation}
  a b = a\circ b - i \left[a,b\right]  
\end{equation}
that justifies all factors 1/2.
\end{rem}
\vspace{2mm}

\section{Infinite-dimensional case}
\label{infcase}

For infinite-dimensional algebras one introduces a Banach structure on the algebra, that is a complete norm $\| \cdot \|$ verifying $\forall\,a,b \in \L$:
\begin{eqnarray}
i) & & \| a \circ b\| \leq \|a\|\ \|b\| ; \\
ii) & & \| a^2 \| = \| a\|^2 ; \\
iii) & & \|a^2\| \leq \|a^2 + b^2\|. 
\label{banach3}
\end{eqnarray} 
In this case the Lie--Jordan algebra is called a Lie--Jordan Banach (LJB) algebra. A key observation helps clarifying the unique correspondence between \LJBs and $\CAs$  \cite{liejordan12}: a $\CA$ is always the complexification of a \LJB and inherits the norm $\| x \| = \| x^*x \|^{1/2}$, where $x = a+ib$ and $x^*x \in \L$.
Hence from an algebraic (and topological) point of view it is completely equivalent to consider $\CAs$ or \LJBs. For instance, commutative $\CAs$ are given by the associative \LJBs, i.e.\ Lie--Jordan algebras where the Jordan product is associative. This can be easily proved.
\begin{theorem}\label{associative LJ}
 A $\CA$ $\cA = \L \oplus i \L$ is commutative if and only if the \LJB $\L$ is associative.
\end{theorem}
\begin{proof}
 Assume first that $\cA$ is commutative. Then, trivially, from the associator identity (\ref{associator}) it follows that the \LJB $\L$ is associative. 
 Conversely, if $\L$ is associative, then any triple commutator vanishes, so that $\forall\, a,b \in \L$
\begin{eqnarray}
0 &=&  [a,[b^2,a]]\\
&=& [a, 2b \circ [b,a]]\\
&=& 2b \circ [a,[b,a]] + [a,2b] \circ [b,a]\\
&=& 2 b \circ [a,[b,a]] - 2 [a,b]^2\\
&=& - 2 [a,b]^2,
\end{eqnarray}
where we used the Leibnitz identity in the second and the third equality.
In conclusion, $[a,b] = 0,\ \forall\, a,b\in \cA$. 
\end{proof}
Classical systems (namely those systems whose observables make up an Abelian algebra)
are therefore characterized by associative LJB--algebras. In the Appendix it is further analyzed how the commutation of the operators can be expressed in terms of the Jordan product only.
In the next section we will define classical and quantum states.

\section{Classical states and quantumness witness}
\label{clstqw}

\subsection{Classical states}
\label{clst}

The space of states $\cS(\L)$ is defined \cite{liejordan12} by all real normalized positive linear functionals on $\L$, i.e. 
\begin{equation}
 \rho\colon \L \to \mathbb{R}                                                                                                                                          
\end{equation}
linear and such that $\rho(\mathds{1}) = 1$ and $\rho(a^2) \geq 0, \ \forall\,a \in \L$.
We can now define classical states.
\begin{definition}
\label{classical state}
We say that a state $\rho \in \cS(\L)$ is \emph{classical} if
\begin{equation}
 \rho( (a\circ b) \circ c - a \circ (b \circ c)) = 0, \qquad \forall\, a,b,c \in \L
 \label{classst}  
\end{equation}
A state that is not classical is \emph{quantum}.
\end{definition}

In other words, classical states do not ``detect" the lack of associativity of the algebra and therefore (nonvanishing) triple commutators.
\begin{rem}
Observe that the above definition 
of classical states is applicable also to  algebras of unbounded operators. More on this after Theorem \ref{classicalfin}. 
\end{rem}

\subsection{Quantumness Witness}
\label{quantumness}

As a consequence of Theorem \ref{associative LJ}, for a quantum (i.e.\ non-commutative) system, it is always possible to find a triple of observables $a,b,c$ such that the observable
\begin{equation}
 q = (a\circ b)\circ c - a \circ (b \circ c)
 \label{qabc} 
\end{equation}
is non-vanishing. Moreover, classical states (\ref{classst}) vanish on $q$-observables.
Thus $q \in \L$ is a candidate ``witness'' for the quantum nature of the algebra of observables. Notice that, unlike usual (entanglement and quantumness) witnesses, $q$ detects quantumness as soon as $q\neq 0$. 
However, if one wants to consider only positive witnesses, one can always use $q^2$ instead of $q$.

\section{Characterization of classical states}
\label{finite}

One can give a nice characterization of classical states. We need some lemmas and  definitions. 

First of all we recall~\cite{Araki} that \emph{normal} states $\rho\in\cS(\L)$ can be \emph{uniquely} realized as traces over density matrices $\tilde{\rho}$ \emph{belonging} to the $\LJB$ $\L$:
\begin{equation}
\rho(a) = \tr (\tilde{\rho} A), \qquad  \tilde{\rho}\in \L, \quad  \tilde{\rho} \geq 0,\quad \tr\tilde{\rho}=1.
\end{equation}
Moreover, the set of all normal states is dense in $\cS(\L)$ in the weak topology. Therefore, considering only normal states is physically equivalent to considering the set of all states~$\cS(\L)$.

\begin{rem}
In the finite dimensional case all states $\rho\in\cS(\L)$ are normal. Following Ref.~\cite{quantumness12} we will freely use  this identification in the finite-dimensional case and will  commit the sin of not distinguishing between states and density matrices.
\end{rem}

\begin{definition}
Given a Lie algebra $\mathfrak{g}$, the derived algebra is $[\mathfrak{g},\mathfrak{g}]$, i.e.\ the subalgebra generated by taking all possible Lie commutators.
\end{definition}

\begin{rem}
In many relevant cases the derived algebra is the whole algebra, and is called ``perfect algebra''. A more stringent result also holds true, that is all semisimple Lie algebras can be generated by repeated commutators of only two elements (see \cite{kuranishi51} and Theorem \ref{semisimple} in Sec.\ \ref{morefin} of this article). This is true, e.g. for the  algebras $\mathfrak{su}(\mathrm{n})$. 
\end{rem}

\begin{lem}
 A normal state $\rho \in \cS(\L)$ is classical if and only if its density matrix $\tilde\rho$ belongs to the center of the derived algebra $[\L,\L]$.
\end{lem}
\begin{proof}
By using the the compatibility condition (\ref{associator}), the Leibniz identity (\ref{leibniz}) and the properties of the trace we have:
\begin{eqnarray}
   \rho( (a\circ b) \circ c - a \circ (b \circ c)) &=& \mathrm{Tr}(\t\rho\circ[b,[c,a]]) \nonumber \\ &=&  \mathrm{Tr}([\t\rho \circ b, [c,a]]) - \mathrm{Tr}(b \circ [\t\rho, [c,a]])\nonumber \\ &=& \mathrm{Tr}(b \circ [\t\rho, [a,c]]). \label{centder}
\end{eqnarray}
Hence $\rho( (a\circ b) \circ c - a \circ (b \circ c)) = 0$ for all $a,b,c \in \L$ implies $[\t\rho, [a,c]] = 0$ for all $a,c \in \L$, i.e.\ $\t\rho$ is in the center of $[\L,\L]$. The converse is obviously true from Eq.\ (\ref{centder}). 
\end{proof}

\begin{lem}
 A density matrix $\t\rho$ is in the center of the algebra $\L$ if and only if the corresponding state $\rho$ satisfies
\begin{equation}
 \rho([a,b]) = 0,
\end{equation}
for all $a,b \in \L$.
\end{lem}
\begin{proof}
If $\t\rho$ is in the center $Z(\L)$ then
\begin{equation}
 \mathrm{Tr}(\t\rho \circ [a,b]) = \mathrm{Tr}([\t\rho \circ a, b]) = 0,
\end{equation}
for all $a,b \in \L$.\\
Conversely, if $\rho([a,b]) = 0$ $\forall\, a,b \in \L$ then
\begin{equation}
 \mathrm{Tr}(\t\rho \circ [a,b]) = \mathrm{Tr}([\t\rho \circ a, b]) - \mathrm{Tr}(a \circ [\t\rho, b]) = - \mathrm{Tr}(a \circ [\t\rho, b]) = 0,
\end{equation}
implies $[\t\rho,b] = 0$ $\forall\, b \in \L$, i.e. $\t\rho \in Z(\L)$.
\end{proof}

From these lemmas it immediately follows the following 
\begin{theorem}
\label{classicalfin} 
 A normal state $\rho$ on a $\LJB$  of observables $\L$ is classical if and only if 
\begin{equation}
\label{charqwfin}
 \rho([a,b]) = 0,
\end{equation}
for all $a,b \in \L$.
\end{theorem}
This recovers, for the finite-dimensional case, the result obtained in Ref.~\cite{quantumness12}.
As emphasized before, classical states are ``transparent" to all commutators and do not ``detect" the non-commutativity of the algebra.

\begin{rem}
Notice that if  $a$ e $a^2 = a \circ a$ commute, then the compatibility relation~(\ref{leibniz}) implies the weak form of associativity~(\ref{jassociativity}). The converse, i.e.\ that weak associativity implies commutativity of $a$ and $a^2$, is not obvious, and is the content of Theorem~\ref{jordancomm} in the Appendix, which relies on the fact that the irreducible representations of a $\CA$ separate its elements.
For general $*$--algebras this might not be true.
\end{rem}

\begin{rem}
Observe that the above characterization (\ref{charqwfin}) would not be suitable in the infinite-dimensional case. Think for example of the CCR, $[\hat x, \hat p] = i \hbar \mathds{1}$\footnote{We stick here to the usual definition, different from (\ref{lie}).}, which would imply $\rho([\hat x, \hat p]) = i \hbar$.
By contrast, Definition \ref{classical state} of classical states is also applicable to infinite-dimensional systems and the use of the Jordan product avoids the problems related to the CCR.
\end{rem}

\section{More on classicality and associativity in Lie--Jordan algebras}
\label{morefin}

By taking advantage of the relations discussed above between $\LJBs$ and other algebraic structures, such as $\CAs$
and Lie algebras, one can obtain new interesting results as direct descendants of older ones. In particular, in this Section we will present two more theorems that can further enlighten our discussion of classicality in the framework of Lie--Jordan algebras.  

The first theorem is on the characterization of associative $\LJBs$:
\begin{theorem}
 Given a $\LJB$ $(\L,\circ,\left[\cdot,\cdot\right])$ with positive cone $\L^+$, the following statements are equivalent:
\begin{enumerate}
\item  $\L$ is associative; \label{classical1} 
\item $\L$ is isomorphic to  $C(X, \mathbb{R})$, for some  locally compact Hausdorff space  $X$; 
if $\L$ has a unit, $X$ is compact; \label{classical2} 
\item if $a,b \in \L^+$ such that $a-b \in \L^+$, then  $a^2 - b^2 \in \L^+$; \label{classical3} 
\item if $a,b \in \L^+$, then $a \circ b \in \L^+$. \label{classical4} 
\end{enumerate}
\end{theorem}
The above theorem translates the ideas explored in Refs.~\cite{Dix77,La98,APVR} in terms of $\LJBs$, their associativity and lack thereof.
It follows that for a quantum system one can always find pairs of observables $a,b \in \L^+$ such that the observable 
\begin{equation}
\label{qAVR}
 q_{\rm AVR} = a \circ b
\end{equation}
is not positive semidefinite and can be adopted as a quantumness witness \cite{AVR,APVR,quantumness12}. This definition appears as a particular case of the more general one (\ref{qabc}).

The second theorem concerns the generation of an algebra from a small set of observables, which is interesting for the following reasons. The characterizations and definitions of classicality (\ref{charqwfin}) and quantumness (\ref{qAVR}) make use of simple commutators and anticommutators (that involve only couple of operators). By contrast, the strategy adopted in this article, hinging upon the identity (\ref{associator}), makes use of commutators and anticommutators that involve \emph{three} operators. If one aims at an operational approach \cite{anticommutators12}, towards experiments, one must make careful use of resources. For example, (traces of) anticommutators involving $n$ operators are related to $n$th-order interference experiments and increasingly complicated quantum circuits.
In light of this observation it is interesting to understand how one can generate the whole algebra by making use of a small set of generators via the Jordan product. We are going to prove in Theorem~\ref{thm:3el} that under suitable hypotheses, three generators are enough.

\begin{definition}
A set of elements $a_1, \ldots, a_k$ is said to \emph{generate} an algebra $\L$ if every element of $\L$ is linearly dependent on products of $a_1, \ldots, a_k$; the elements $a_1, \ldots, a_k$ are then called \emph{generators} of $\L$.
\end{definition}
Kuranishi proved a sufficient condition for Lie algebras to be generated by two elements \cite{kuranishi51}:

\begin{theorem}\label{semisimple}
 Let $\mathfrak{g}$ be a semisimple Lie algebra over the real or complex numbers. Then there exist two elements $a$ and $b$ which generate $\mathfrak{g}$.
\end{theorem}

\begin{definition}
 A non-unital Lie--Jordan algebra $(\L,\circ,\left[\cdot,\cdot\right])$ is called semisimple if the Lie algebra $(\L,\left[\cdot,\cdot\right])$ (i.e.\ the full algebra considered with the Lie product alone) is semisimple.
 \end{definition}

\begin{rem} 
Observe that the unit cannot be generated by Lie products, and hence it will be ``added by hand'' whenever necessary.
\end{rem}

The analogous for Lie--Jordan algebras  of Kuranishi's theorem is the following
\begin{theorem}
 Let $(\L,\circ,\left[\cdot,\cdot\right])$ be a semisimple Lie--Jordan algebra. Then the Jordan algebra $(\L,\circ)$ (i.e. the full algebra considered with the Jordan product alone) is generated by the Jordan products of three elements, plus the identity. In particular, one can use two generators $a$, $b$ of the Lie algebra $(\L,\left[\cdot,\cdot\right])$, and their commutator $c = [a,b]$.
\label{thm:3el}
\end{theorem}
\begin{proof}
 For simplicity assume, without loss of generality, that the algebra is non-unital. Since the Lie algebra $(\L,\left[\cdot,\cdot\right])$ is semisimple it can be generated by repeated Lie products of two elements $a$ and $b$. Then starting from $a$ and $b$ we generate with a first Lie bracket: 
\begin{equation}
 c = [a,b]
\end{equation}
then
\begin{equation}
 [a,c] = [a,[a,b]], \quad [b,c] = [b,[a,b]]
\end{equation}
and repeating 
\begin{equation}
 [a,[a,c]], [a,[b,c]], [b,[a,c]],[b,[a,c]],[b,[b,c]],[c,[a,c]],[c,[b,c]]
\end{equation}
and so on. We see that all the elements generated by $a$, $b$ and $c$ are of the form of a triple commutator. Recalling the associator identity (\ref{associator}) 
\begin{equation}
 (a\circ b)\circ c - a \circ (b \circ c) = \left[b,\left[c,a\right]\right],
\end{equation}
it follows that every element can be expressed as a linear combination of triple Jordan products, that is generated by Jordan products of $a,b$ and $c = [a,b]$.\\
\qedhere
\end{proof}

For finite-dimensional quantum systems, the space of quantum states can be immersed into the semisimple Lie--Jordan algebra $\mathfrak{u}(\mathrm{N})$, which can be Jordan-generated by three elements.
This means that one could witness the properties of a system by repeated measures of anticommutators of appropriate elements of the algebra, which is in principle experimentally feasible. Further investigation is required to check if the elements generating the algebra correspond to realization of states as projectors in the algebra.

\section{Conclusions and perspectives}
\label{concl}

In this article we  proposed an alternative definition of quantumness witnesses and classical states that directly refers to the algebra of  observables, and showed its equivalence with previous definitions that refer instead to the whole associative algebra, which is assumed to embed the real algebra of observables.  
This new definition has the advantage of being more operational, since it makes use only of observables. Moreover, it is also more suitable for generalization towards algebras of unbounded operators, because the Jordan product avoids the problems related to the CCR.

The  characterization of classicality and quantumness adopted in this article involves \emph{three} observables, in contrast to the use of a simple commutator. This could have a relation to the interesting property, proven in the paper, that  
semisimple Lie--Jordan algebras are fully generated by three appropriate observables, a results that points towards experiments.

As emphasized at the end of Sec.\ \ref{finite}, while the characterization (\ref{charqwfin}) is not suitable in the unbounded case, Definition \ref{classical state} in Sec.\ \ref{clstqw} still applies.
This opens a door to the study of unbounded observables in infinite-dimensional systems, such as coherent states \cite{sudarshan,glauber}, superpositions thereof, and the semiclassical limit \cite{FM1,FM2}. 

We conclude with an observation. It is clear that the notion of classicality is not necessarily related to that of macroscopicity or to the thermodynamical limit. The research of the last few decades has shown that a quantum system can behave classically (or semiclassically) if (some of) its observables commute  with every element in the algebra \cite{holevo}. For example, dissipative quantum systems display a number of classical features \cite{CFMP}. The notions of classicality and quantumness should then be contrasted on the basis of the different footings on which observables and states stand. A natural conclusion is that one should look at the non-commutativity of the algebras. A somewhat more hidden aspect is that the lack of associativity may play an important role. This is the message we tried to convey in this article.

\ack 
This work was partially supported by PRIN 2010LLKJBX on ``Collective quantum phenomena: from strongly correlated systems to quantum simulators". LF would like to thank F. Falceto for his kind hospitality at the University of Zaragoza where part of this work was prepared.

\appendix

\section{Commutation from the Jordan product}
\label{comm}

Consider two elements $a$ and $b$ of a Lie--Jordan Banach algebra $\L$. We may ask: which properties of the elements $a$ and $b$ can be expressed in terms of the Jordan product only, corresponding to the commutation relation $[a,b] = 0$? Of course the Jordan product is commutative but need not be associative. Following \cite{Hanche84} we will see which condition on the Jordan product is equivalent to the commutation of two elements.
Let us denote the Jordan multiplier by an element $a$ as $\delta_a$. Thus for all $c \in \L$ 
\begin{equation}
 \delta_a(c) = a \circ c = \frac{1}{2}\{a,c\}.
\end{equation}
\begin{definition}[Jordan-commutation]
 Two elements $a$ and $b$ are said to \emph{Jordan-commute} if the operators $\delta_a$ and $\delta_b$ commute:
\begin{eqnarray}\label{jordan comm}
 (\delta_a \delta_b - \delta_b \delta_a) (c) &=& a \circ (b \circ c) - b \circ (a \circ c) = 0, \qquad \forall\, c \in \L\\
&=& \frac{1}{4} \{a,\{b,c\}\} - \frac{1}{4} \{b,\{a,c\}\} = 0, \qquad \forall\, c \in \L.
\end{eqnarray}
\end{definition}
\begin{theorem}\label{jordancomm} Two elements $a$ and $b$ of a \LJB algebra $\L$ Jordan-commute if and only if they commute in the usual sense, i.e. $[a,b] = 0$.
\end{theorem}
\begin{proof}
 From the associator identity (\ref{associator}), $a$ and $b$ Jordan-commute by definition (\ref{jordan comm}) if and only if
 \begin{equation}\label{ass id comm}
   a \circ (b \circ c) - b \circ (a \circ c) = [[a,b],c] = 0, \qquad \forall\, c \in \L.
 \end{equation}
If $[a,b] = 0$ then they trivially Jordan-commute. To prove the converse, recall that a \LJB algebra is the self-adjoint part of a C$^*$--algebra \cite{liejordan12}, and that every irreducible representation of a C$^*$--algebra separates elements, i.e.\ for every element of the algebra such that $x \neq 0$ there exists an irreducible representation $\pi$ such that $\pi(x) \neq 0$. Then assume $a$ and $b$ Jordan-commute and hence from Eq.\ (\ref{ass id comm}) their commutator $[a,b]$ is in the Lie center of the algebra $\L$. Assume by absurdum that $[a,b] \neq 0$ and let $\pi$ be an irreducible representation such that $\pi([a,b]) \neq 0$. Then from Schur's lemma it is well known that the Lie center of an irreducible representation is a multiple of the identity and this implies by antisimmetry that $\pi([a,b]) = 0$, which is a contradiction.
\end{proof}

\end{document}